\documentclass[draftclsnofoot,onecolumn,12pt]{IEEEtran}

\usepackage[T1]{fontenc}
\usepackage{amsthm}
\usepackage{amsmath}
\usepackage[cmintegrals]{newtxmath}
\usepackage{mathtools,bm}
\usepackage{graphicx}
\usepackage{color}
\usepackage{multicol}
\usepackage[font={footnotesize}]{caption}

\usepackage{lettrine}
\usepackage{footnote}
\makesavenoteenv{tabular}
\usepackage[noend]{algpseudocode}

\makeatletter
\newcommand{\Vast}{\bBigg@{5}}
\newcommand{\vast}{\bBigg@{3}}
\makeatother

\renewcommand\IEEEkeywordsname{Index Terms}
\renewcommand{\IEEEQED}{$\Box$}

\newcommand{\km}{\,\mathrm{km}}
\newcommand{\kmkm}{\,\mathrm{km^2}}
\newcommand{\W}{\,\mathrm{W}}
\newcommand{\dd}{\,\mathrm{d}}
\newcommand{\dB}{\,\mathrm{dB}}
\newcommand{\maxi}{\,\mathrm{maximize}}

\newcommand{\BSkm}{\,\mathrm{BSs/km^2}}
\newcommand{\Userkm}{\,\mathrm{users/km^2}}
\newcommand{\opt}{\,\mathrm{opt}}
\newcommand{\s}{\,\mathrm{s}}

\theoremstyle{definition}
\newtheorem{theo}{\textbf{Theorem}}

\newtheorem{cor}{\textbf{Corollary}}
\newtheorem{rem}{\textit{Remark}}

\usepackage{siunitx}
\makeatletter
\providecommand\add@text{}
\newcommand\tagaddtext[1]{%
  \gdef\add@text{#1\gdef\add@text{}}}%
\renewcommand\tagform@[1]{%
  \maketag@@@{\llap{\add@text\quad}(\ignorespaces#1\unskip\@@italiccorr)}%
}
\makeatother

\begin{document}

\title{Traffic Offloading Probability for Integrated LEO Satellite-Terrestrial Networks}

\author{ Hossein~Akhlaghpasand, and~Vahid~Shah-Mansouri
\thanks{Hossein Akhlaghpasand and Vahid Shah-Mansouri are with the School of Electrical and Computer Engineering, University of Tehran, Tehran, Iran (e-mail: hakhlaghpasand@ut.ac.ir, vmansouri@ut.ac.ir).}}

\maketitle

\begin{abstract}
In this paper, we consider traffic offloading of integrated low earth orbit (LEO) satellite-terrestrial network. We first derive traffic offloading probability from the terrestrial network to the LEO satellite network based on the instantaneous radio signal strength. Then to overcome limited coverage and also traffic congestion of the terrestrial network, we design an optimal satellite network in which the minimum number of LEO satellites maximizes the traffic offloading probability, while probability of a generic LEO satellite being idle is close to zero. Since the satellite network is optimized regarding the intensities of the base stations (BSs) and the users, it can control the terrestrial traffic. Numerical results show that an appropriate number of LEO satellites overcomes the limited coverage in a region with low intensity of the BSs and also the traffic congestion by controlling the traffic in a region that the intensity of the users increases.
\end{abstract}


\section{Introduction}
To expand a wide range of user services in the next generation of the wireless systems, LEO satellite network has received significant attention by providing extensive communications all over the Earth \cite{r1}-\cite{r3}. Recently, data traffic of the terrestrial network considerably increases. Traffic offloading from the terrestrial network to the satellite network becomes a solution to manage the increasing data traffic \cite{r4}.

Traffic offloading protocols in \cite{r5} are considered for the integration of satellite networks into existing terrestrial network infrastructures, as well as in \cite{r6} the authors adapt the offloading problem to the predicted traffic. In \cite{r7}, the traffic offloading and the energy efficiency are optimized  through cooperation between the satellite and the terrestrial networks. Du \emph{et al.} in \cite{r8} design a second-price auction mechanism in which the terrestrial network sells the channel to the beam groups of the satellite. The authors in \cite{r9} suggest a new mechanism based on Stackelberg game for pricing data offloading between the satellite and the terrestrial networks. In \cite{r10}, low latency data is transferred through the terrestrial network, while broadband data is transferred through the satellite network. The rate of each terrestrial cell is maximized in \cite{r11} such that the power consumption of the satellite in the backhaul link is low.

Our motivation is to propose a framework based on the traffic offloading probability for design of a LEO satellite network in order to reinforce the terrestrial network by overcoming the limited coverage and also the traffic congestion. The main contributions of this paper can be summarized as follows:
\begin{itemize}
  \item We introduce an integrated satellite-terrestrial network in which each user may be serviced by either a terrestrial network with Poisson point process (PPP) distribution for the BSs, or a satellite network with Binomial point process (BPP) distribution for the LEO satellites.
  \item We derive traffic offloading probability from the terrestrial network to the satellite network based on the instantaneous radio signal strength in which the satellite channel is modeled with respect to its high-speed movement. We also compute probability of a generic LEO satellite being idle which is called \emph{satellite-empty probability}. For BPP model, these probabilities have not been considered in prior works.
  \item We suggest a mathematical problem for design of an optimal satellite network to control the terrestrial traffic, in which the minimum number of LEO satellites (as a control parameter for reinforcing the terrestrial network) maximizes the traffic offloading probability such that the satellite-empty probability is close to zero.
\end{itemize}
Numerical results demonstrate that an appropriate number of LEO satellites reinforces the terrestrial network. Our experiments also show how important parameters such as the altitude of the satellite network and the radiant power of the satellites affect the traffic offloading probability. 

\section{Problem Setup}
We consider an integrated satellite-terrestrial network consisting of $N$ LEO satellites distributed as BPP on a spherical surface around the center of the Earth and at the altitude of ${r}_s$ above the Earth; and multiple BSs distributed as PPP with intensity $B$ in a certain area on the ground. Since PPP can fit the ground network well in performance analysis, we exploit it for distribution of the BSs. But, PPP is not suitable for modeling a finite area network with limited nodes, e. g., the LEO satellite network. For such cases, BPP is an appropriate model to capture the characteristics of the network \cite{r12}. The radius of the Earth is denoted by ${r}_e$. Since the terrestrial and the satellite networks are independent, the overall network can be considered as a superposition of two independent Voronoi tessellations, one for the LEO satellites and the other for the BSs. According to the instantaneous radio signal strength, each user may be serviced by either the nearest BS at the terrestrial network or the nearest LEO satellite at the satellite network. It is worthwhile noting that the BSs are not uniformly distributed all over the Earth. So, we denote $B$ for a certain region that the intensity is almost constant. Similarly the users are spatially distributed in a certain region according to an independent PPP with intensity $\mathcal{U}$. We also consider a general power loss model in which the power of the signal decreases as $r^{-\eta}$ with the distance $r$, where $\eta$ is the path loss exponent.

The channel between the test user and LEO satellite is denoted by ${h}_s$ and modeled as Rice-Rayleigh/lognormal fading in which the Rician fading represents unshadowed areas with high received power (good state) and Rayleigh/lognormal fading represents shadowed areas with low received power (bad state) \cite{r13}. By taking into account the high-speed movement of the LEO satellites, the probability distribution function (pdf) of the Rice-Rayleigh/lognormal fading power is given by
\begin{align}\label{eq1}
{f_{\left| h_s \right|^2}}\left(h;t\right) = \left(1-{P_f}\left(t\right)\right){K}\left(t\right) e^{-{K}\left(t\right)\Big(h+1\Big)}{I_0}\left(2{K}\left(t\right)\sqrt{h}\right) \nonumber \\
+ \frac{10{P_f}\left(t\right)}{\sqrt{2\pi}{\varsigma}\left(t\right) {\ln} \left(10\right)} \int_{0}^{+\infty} \frac{1}{h_0^2}e^{-\left(\frac{h}{h_0}+\frac{\left(10\log_{10}\left(h_0\right)-{\mu}\left(t\right)\right)^2}{2{\varsigma^2}\left(t\right)}\right)} \dd h_0 ,
\end{align}
for $h \geq 0$, where ${P_f}$, ${K}$, ${\mu}$ and ${\varsigma}$ depend on the elevation angle of the LEO satellite and respectively represent the bad state probability of the channel, the Rice factor, the mean power level decrease and the variance of the power level due to shadowing. ${I_0}\left(\cdot\right)$ is also the modified Bessel function of order zero. When a LEO satellite passes from horizon to horizon, the rapid change of the elevation angle is distinguished according to the time, $t$. Hence, the elevation-dependent parameters can be expressed as the functions of $t$. Moreover, due to BPP model of the satellite network, the pdf of the user's distance from the nearest LEO satellite, $R_s$, is \cite{r12}
\begin{equation}\label{eq2}
{f_{R_s}}\left(r\right) = N \left(1-\frac{r^2-{r}_s^2}{4{r}_e \left({r}_e+{r}_s\right)}\right)^{N-1} \frac{r}{2{r}_e \left({r}_e+{r}_s\right)} , \quad {r}_s \leq r \leq 2{r}_e+{r}_s .
\end{equation}
The channel between the test user and BS is denoted by $h_b$ which is modeled as Rayleigh fading. The cumulative distribution function (CDF) of the Rayleigh fading power is given by
\begin{equation}\label{eq3}
{F_{\left|h_b\right|^2}}\left(h\right) \triangleq \Pr \left(\left|h_b\right|^2 \leq h\right) = 1-e^{-\frac{h}{2\sigma^2}} , \quad h \geq 0 ,
\end{equation}
where $\sigma$ is the parameter of the Rayleigh fading. Finally, due to PPP model of the terrestrial network, the CDF of the user's distance from the nearest BS, $R_b$, is obtained as
\begin{equation}\label{eq4}
{F_{R_b}}\left(r\right) = 1-e^{-B \pi r^2} , \quad r \geq 0 .
\end{equation}

\section{Traffic Offloading Probability}
In this section, we compute traffic offloading probability of the test user from the terrestrial network to the satellite network. According to Slivnyak's theorem, the analysis proposed here for the test user holds for any generic user.
\begin{theo}\label{theo1}
In the integrated satellite-terrestrial network with Rice-Rayleigh/lognormal fading for the satellite links and Rayleigh fading for the terrestrial links, by equality of the path loss exponents for two links, traffic offloading probability of the test user from the nearest BS to the nearest satellite is obtained as
\begin{align}\label{eq5} 
{P}_s = 2\sigma^2 N \int_{0}^{+\infty} \frac{e^{-\left(\pi B\left(\frac{{\mathcal{P}}_s}{\mathcal{P}_b}h\right)^{-\frac{2}{\eta}}{r}_s^2+{\mathcal{D}}\left(h\right)\right)}{\gamma}\left(N,-{\mathcal{D}}\left(h\right)\right)}{\left(-{\mathcal{D}}\left(h\right)\right)^N} \Bigg[{P_f}\left(t\right){\mathcal{Q}}\left(h;t\right) \nonumber \\
+\left(1-{P_f}\left(t\right)\right) e^{-{K}\left(t\right)}\sum_{z=1}^{+\infty}\frac{z {K^{2z-1}}\left(t\right)\left(2\sigma^2 h \right)^{z-1}}{\left(z-1\right) ! \left(2\sigma^2 h{K}\left(t\right)+1\right)^{z+1}}\Bigg] \dd h ,
\end{align}
where ${\mathcal{D}}\left(h\right)=4\pi B \left({\mathcal{P}}_sh/\mathcal{P}_b\right)^{-2/\eta}{r}_e\left({r}_e+{r}_s\right)$ and
\begin{equation}
{\mathcal{Q}}\left(h;t\right)=\frac{10}{\sqrt{2\pi}{\varsigma}\left(t\right){\ln}\left(10\right)}\int_{0}^{+\infty}\frac{e^{-\left(10\log_{10}\left(h_0\right)-{\mu}\left(t\right)\right)^2/2{\varsigma^2}\left(t\right)}}{\left(h_0+2\sigma^2 h\right)^2} \dd h_0 .
\end{equation}
${\mathcal{P}}_s$ and $\mathcal{P}_b$ respectively represent the radiant powers of the satellites and the BSs. ${\gamma }\left(\cdot,\cdot\right)$ is the lower incomplete gamma function.
\end{theo}
\begin{proof}
For servicing the test user, the network selection criteria is the radio signal strength, therefore the probability that the test user is associated to the satellite can be written as
\begin{align}\label{eq6}
{P}_s &= \Pr \left\{{\mathcal{P}}_s\left| h_s\right|^2 R_s^{-\eta} \geq \mathcal{P}_b\left| h_b\right|^2 R_b^{-\eta} \right\} \nonumber \\
&= \Pr \left\{\left(\frac{R_s}{R_b}\right)^{\eta} \leq \left(\frac{{\mathcal{P}}_s}{\mathcal{P}_b} \cdot \frac{\left| h_s\right|^2}{\left| h_b\right|^2} \right) \right\} \nonumber \\
&= \int_{h} {F_{R_{s/b}}}\left(\frac{{\mathcal{P}}_s}{\mathcal{P}_b}h\right) {f_{\left| h_{s/b}\right|^2}}\left(h;t\right) \dd h ,
\end{align}
where $R_s$ and $R_b$ respectively represent the user's distance from the nearest satellite and the nearest BS. For simplicity, we define two variables $R_{s/b} \triangleq \left(\frac{R_s}{R_b}\right)^{\eta}$ and $\left|h_{s/b}\right|^2 \triangleq \frac{\left|h_s\right|^2}{\left|h_b\right|^2}$. The CDF of $ R_{s/b}$ is obtained as
\begin{align}\label{eq7}
{F_{R_{s/b}}}\left(r\right) &= \Pr \left\{\left(\frac{R_s}{R_b}\right)^{\eta} \leq r \right\} \nonumber \\
&= 1-\int_{{r}_s}^{2{r}_e+{r}_s} {F_{R_{b}}}\left(r^{-\frac{1}{\eta}}r_0\right) {f_{R_s}}\left(r_0\right) \dd r_0 \nonumber \\
&= Ne^{-\left(\pi B r^{-\frac{2}{\eta}}{r}_s^2+{\mathcal{D}}\left(\frac{{\mathcal{P}}_b}{{\mathcal{P}}_s}r\right)\right)} \int_{0}^{1} x^{N-1}e^{{\mathcal{D}}\left(\frac{{\mathcal{P}}_b}{{\mathcal{P}}_s}r\right)x} \dd x \nonumber \\
&= \frac{Ne^{-\left(\pi B r^{-\frac{2}{\eta}}{r}_s^2+{\mathcal{D}}\left(\frac{{\mathcal{P}}_b}{{\mathcal{P}}_s}r\right)\right)}{\gamma}\left(N,-{\mathcal{D}}\left(\frac{{\mathcal{P}}_b}{{\mathcal{P}}_s}r\right)\right)}{\left(-{\mathcal{D}}\left(\frac{{\mathcal{P}}_b}{{\mathcal{P}}_s}r\right) \right)^N} ,
\end{align}
where we used \eqref{eq2}, \eqref{eq4} and the substitution $x = 1-\frac{r_0^2-{r}_s^2}{4{r}_e \left({r}_e+{r}_s\right)}$ in the third equality, and the definition of the incomplete gamma function in the fourth equality. The CDF of $\left|h_{s/b}\right|^2$ is
\begin{align}\label{eq8} 
{F_{\left| h_{s/b}\right|^2}}\left(h;t\right) &= \Pr \left\{\frac{\left| h_{s}\right|^2}{\left| h_{b}\right|^2} \leq h \right\} \nonumber \\
&=1-\int_{0}^{+\infty} {F_{\left| h_{b}\right|^2}}\left(\frac{h'}{h}\right) {f_{\left| h_{s}\right|^2}}\left(h';t\right) \dd h' \nonumber \\
&= \frac{10{P_f}\left(t\right)}{\sqrt{2\pi}{\varsigma}\left(t\right){\ln}\left(10\right)} \int_{0}^{+\infty}\frac{2\sigma^2 he^{-\frac{\left(10\log_{10}\left(h_0\right)-{\mu}\left(t\right)\right)^2}{2{\varsigma^2}\left(t\right)}}}{h_0\left(h_0+2\sigma^2 h\right)} \dd h_0 \nonumber \\
&~~~ +\left(1-{P_f}\left(t\right)\right) e^{-{K}\left(t\right)} \sum_{z=1}^{+\infty} \frac{{K^{2z-1}}\left(t\right)}{\left(z-1\right)!}\left(\frac{2\sigma^2 h}{2\sigma^2 h{K}\left(t\right)+1}\right)^z ,
\end{align}
where we utilized \eqref{eq1}, \eqref{eq3} and the series expansion of $I_0$ in the third equality. By differentiating \eqref{eq8}, we have
\begin{align}\label{eq9}
{f_{\left| h_{s/b}\right|^2}}\left(h;t\right) = \frac{\dd}{\dd h} {F_{\left| h_{s/b}\right|^2}}\left(h;t\right) = 2\sigma^2\Biggl(\frac{10{P_f}\left(t\right)}{\sqrt{2\pi}{\varsigma}\left(t\right){\ln}\left(10\right)}\int_{0}^{+\infty}\frac{e^{-\frac{\left(10\log_{10}\left(h_0\right)-{\mu}\left(t\right)\right)^2}{2{\varsigma^2}\left(t\right)}}}{\left(h_0+2\sigma^2 h\right)^2} \dd h_0 \nonumber \\
+\left(1-{P_f}\left(t\right)\right)  e^{-{K}\left(t\right)} \sum_{z=1}^{+\infty} \frac{z {K^{2z-1}}\left(t\right) \left(2\sigma^2 h\right)^{z-1}}{\left(z-1\right) ! \left(2\sigma^2 h{K}\left(t\right)+1\right)^{z+1}}\Biggr) .
\end{align}
The proof is complete with replacing \eqref{eq7} and \eqref{eq9} in \eqref{eq6}.
\end{proof}
\begin{cor}\label{cor1} 
${P}_s$ is increasing function of $N$.
\end{cor}
\begin{proof}
Regarding \eqref{eq5}, the part of $P_s$ which is proportional to $N$ is denoted by ${\mathcal{G}}\left(N;h\right) = \frac{N {\gamma}\left(N,-{\mathcal{D}}\left(h\right)\right)}{\left(-{\mathcal{D}}\left(h\right)\right)^N}$, and rewritten as ${\mathcal{G}}\left(N;h\right) = \frac{e^{{\mathcal{D}}\left(h\right)}}{1+{\mathcal{D}}\left(h\right)/{\mathcal{W}}\left(N;h\right)}$ \cite{r14}, where
\begin{equation}\label{eq12}
{\mathcal{W}}\left(N;h\right) \triangleq N+1-\frac{{\mathcal{D}}\left(h\right)}{N+2+\frac{\left(N+1\right){\mathcal{D}}\left(h\right)}{N+3-\frac{2{\mathcal{D}}\left(h\right)}{N+4+\cdots}}} .
\end{equation}
The first term in ${\mathcal{W}}\left(N;h\right)$ increases with respect to $N$, while the third term decreases (${\mathcal{D}}\left(h\right) \geq 0$). Hence, ${\mathcal{W}}\left(N;h\right)$ and consequently ${\mathcal{G}}\left(N;h\right)$ are increasing functions of $N$.
\end{proof}
Regarding Corollary \ref{cor1}, we interpret that increasing number of LEO satellites reduces the terrestrial traffic, however it may cause that some LEO satellites remain idle. For design of an optimal satellite network in which no LEO satellite is idle, we exploit a condition based on the satellite-empty probability.
\begin{theo}\label{theo2} 
The satellite-empty probability of a generic LEO satellite for $N \geq c\left({r}_e+{r}_s\right)^2\mathcal{U}_s$ is approximately given by
\begin{equation}\label{eq13}
{f_{U}}\left(u=0\right) = 1-\frac{c\left({r}_e+{r}_s\right)^2\mathcal{U}_s}{N} ,
\end{equation}
where $c = 3.6$ is a constant for the two-dimensional Voronoi tessellation and $\mathcal{U}_s={P}_s\mathcal{U}$.
\end{theo}
\begin{proof} 
The area distribution of a satellite region $A$ is given by ${f_{A}} \left(a\right) \approx \frac{{\beta} \left(a/\left({r}_e+{r}_s\right)^2;c+1/2,N-c \right)}{\left({r}_e+{r}_s\right)\sqrt{a}}$, $0 \leq a \leq \left({r}_e+{r}_s\right)^2$ for BPP model, where $\beta\left(x;\varepsilon,\zeta\right) \triangleq \frac{\Gamma\left(\varepsilon+\zeta\right)x^{\varepsilon-1}\left(1-x\right)^{\zeta-1}}{\left(\Gamma\left(\varepsilon\right)\Gamma\left(\zeta\right)\right)}$ and $\Gamma\left(\cdot\right)$ is the gamma function \cite{r15,r16}. Since the users associated to the satellite are distributed as PPP with the intensity $\mathcal{U}_s$, we have $\Pr\left\{U=u|A=a\right\} = \frac{\left(\mathcal{U}_s a\right)^ue^{-\mathcal{U}_s a}}{u!}$ for given $A=a$, so the probability of no user being in $A$ is obtained as
\begin{align}\label{eq15}
{f_{U}}\left(u=0\right) &= \int_{0}^{\left({r}_e+{r}_s\right)^2} \Pr \left\{U=0 | A=a \right\} {f_{A}} \left(a\right) \dd a \nonumber \\
&= \frac{{\Gamma}\left(N\right)}{{\Gamma}\left(c\right){\Gamma}\left(N-c\right)\left({r}_e+{r}_s\right)^2} \int_{0}^{\left({r}_e+{r}_s\right)^2}\left(\frac{a}{\left({r}_e+{r}_s\right)^2}\right)^{c-1} \left(1-\frac{a}{\left({r}_e+{r}_s\right)^2}\right)^{N-c-1} e^{-\mathcal{U}_s a} \dd a \nonumber \\
&= 1-\frac{\left({r}_e+{r}_s\right)^2\mathcal{U}_s c}{N}+\frac{\left({r}_e+{r}_s\right)^4\mathcal{U}_s^2 c\left(c+1\right)}{2N\left(N+1\right)}-\cdots ,
\end{align}
where we utilized Maclaurin series of $e^{-\mathcal{U}_s a}$. For $ N \geq c\left({r}_e+{r}_s\right)^2\mathcal{U}_s$, the third term and later are ignored.
\end{proof}
According to the above description, for design of an optimal satellite network, we can obtain ${N}^{\opt}$ by solving
\begin{align}\label{eq16} 
& \underset{N}{\maxi}
~~~{P}_s \nonumber \\
& \text{subject to} ~~~ {f_{U}}\left(u=0\right) \leq \epsilon ,
\end{align}
where $\epsilon$ is a threshold greater than (but close to) zero. To solve \eqref{eq16}, we present Corollary \ref{cor2} in the following.
\begin{cor}\label{cor2} 
The numerator order of the fraction in \eqref{eq13} is lower than linear with respect to $N$.
\end{cor}
\begin{proof}
In the numerator of the fraction, $\mathcal{U}_s$ is proportional to $N$ as ${\mathcal{G}}\left(N;h\right)$. Since there is ${\mathcal{W}}\left(N;h\right)$ in both the numerator and the denominator of ${\mathcal{G}}\left(N;h\right)$, so ${\mathcal{G}}\left(N;h\right)$ increases with the order lower than one (linear) as $N$ increases.
\end{proof}
Based on Corollaries \ref{cor1} and \ref{cor2}, both ${P}_s$ and ${f_{U}}\left(u=0\right)$ are increasing functions of $N$. Hence, ${P}_s^{\opt}$ (the maximum of ${P}_s$) in \eqref{eq16} is obtained for ${f_{U}}\left(u=0\right)=\epsilon$. We first assume that $N$ is a real variable (not integer) and after solving \eqref{eq16}, if the found ${N}^{\opt}$ is not integer, we round it to the nearest integer greater than that value, until we are sure to reach ${P}_s^{\opt}$.
\begin{figure}
\centering
\includegraphics[width=0.45\columnwidth]{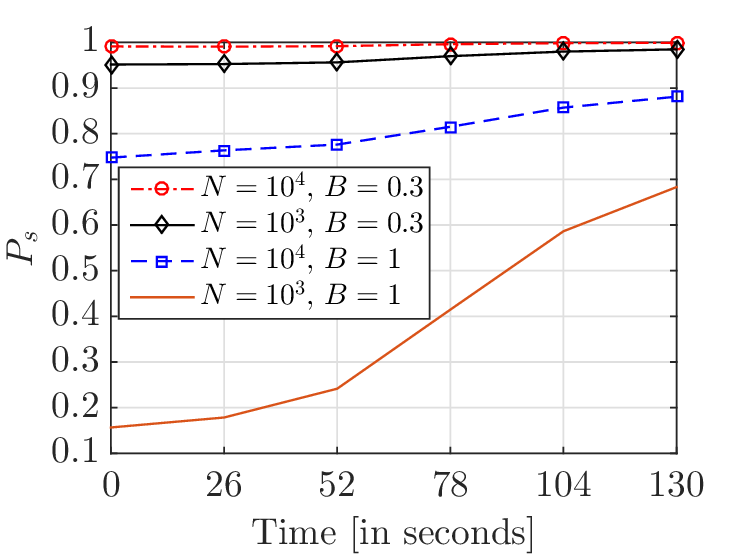}
\caption{Offloading probability vs. the time for ${\mathcal{P}}_s=8~\W$ and ${r}_s=500~\km$.}
\label{fig1}
\end{figure}

Traffic congestion of the terrestrial network can be controlled by solving \eqref{eq16} and obtaining ${P}_s^{\opt}$, such that the ratio ${P}_s^{\opt}$ of the traffic (or the intensity of the users $\mathcal{U}$) is offloaded by the satellite network, i. e., ${\mathcal{U}}_s={P}_s^{\opt}\mathcal{U}$.
\begin{rem}\label{rem1}
Since $B$ and $\mathcal{U}$ are not uniform all over the Earth, for a region with the given $B$ and $\mathcal{U}$, we locally set the number of LEO satellites (proportional to the area of that region), such that the density of the satellite network becomes $\frac{{N}^{\opt}}{4\pi\left({r}_e+{r}_s\right)^2}$.
\end{rem}

\section{Numerical Results}
We numerically evaluate traffic offloading probability and study how it depends on the network parameters. We set ${r}_e=6378~\km$, $\mathcal{P}_b=1~\W$, $\sigma=4.47 \times 10^{-7}$, and $\eta=3$.
\begin{table}[!t]
\begin{center}
\caption{Satellite channel parameters for ${r}_s=500~\km$ \cite{r13}.}
\label{TT}
\begin{tabular}{||c|c|c|c|c||}
\hline
Time, $t$ [s] & ${P_f}$ & ${K}$ & ${\mu}$ [dB] & ${\varsigma}$ [dB] \\ \hline \hline
$0$ & $0.82$ & $3.1$ & $-16$ & $5$ \\ \hline
$26$ & $0.79$ & $3.2$ & $-14$ & $5.5$ \\ \hline
$52$ & $0.69$ & $3.7$ & $-9$ & $4.7$ \\ \hline
$78$ & $0.51$ & $5$ & $-8.6$ & $3.1$ \\ \hline
$104$ & $0.35$ & $6.2$ & $-6.1$ & $1.2$ \\ \hline
$130$ & $0.27$ & $7.3$ & $-3.5$ & $0.2$ \\ \hline
\end{tabular}
\end{center}
\end{table}

In Fig. \ref{fig1}, ${P}_s$ is depicted versus the time interval that the elevation angle of the satellite at ${r}_s=500~\km$ changes from $10^{\circ}$ (corresponding $t=0$) to $60^{\circ}$ (corresponding $t=130~\s$). The parameters of the satellite channel change through the time according to the values presented in Table \ref{TT} \cite{r13}. The curves are plotted for ${\mathcal{P}}_s=8~\W$. From this figure, $P_s$ increases by increasing the satellite elevation angle, because the satellite channel gradually improves from the bad state to the good state. We also obtain $P_s \simeq 1$ even with fewer satellites $N=10^3$ when the intensity of the BSs is low (e. g., $B=0.3~\BSkm$). This means that we can overcome the limited coverage for the low intensity of the BSs, by choosing an appropriate number of satellites. Moreover, increasing $N$ from $10^3$ to $10^4$ can compensate the low probability of the offloading for the high intensity of the BSs $B=1~\BSkm$.

Fig. \ref{fig2} illustrates the impact of the altitude and the radiant power of the satellites on the design of the optimal satellite network. We plot the optimal offloading probability versus the optimal number of LEO satellites for ${r}_s\in\left[600,1000\right]~\km$, $B= 0.5~\BSkm$ and an area $750~\kmkm$. The satellite channel parameters are set to ${P_f}=0.27$, ${K}=7.3$, ${\mu}=-3.5~\dB$, and ${\varsigma}=0.2~\dB$. Comparison between the points P1 and P2 in Fig. \ref{fig2} indicates that ${P}_s^{\opt}$ increases by reducing ${r}_s$ from $900~\km$ to $600~\km$. Also, comparing P3 with P5 shows that ${P}_s^{\opt}$ increases by increasing ${\mathcal{P}}_s$ from $3~\W$ to $5~\W$. Moreover, comparison between P5 and P6 demonstrates that the number of LEO satellites being idle decreases by reducing $\epsilon$ from $0.1$ to $0.01$. Finally to show how to relieve the traffic congestion of the terrestrial network, we display P1 and P4 in Fig. \ref{fig2} where $P_s^{\opt} \approx 0.53$. Since $P_s^{\opt}$ is close to $0.5$, half of the traffic in the terrestrial network is offloaded by the satellite network. While $\mathcal{U}$ increases from $1~\Userkm$ (P4) to $3~\Userkm$ (P1), $P_s^{\opt}$ is almost constant, so increasing the traffic can be controlled.
\begin{figure}
\centering
\includegraphics[width=0.50\textwidth]{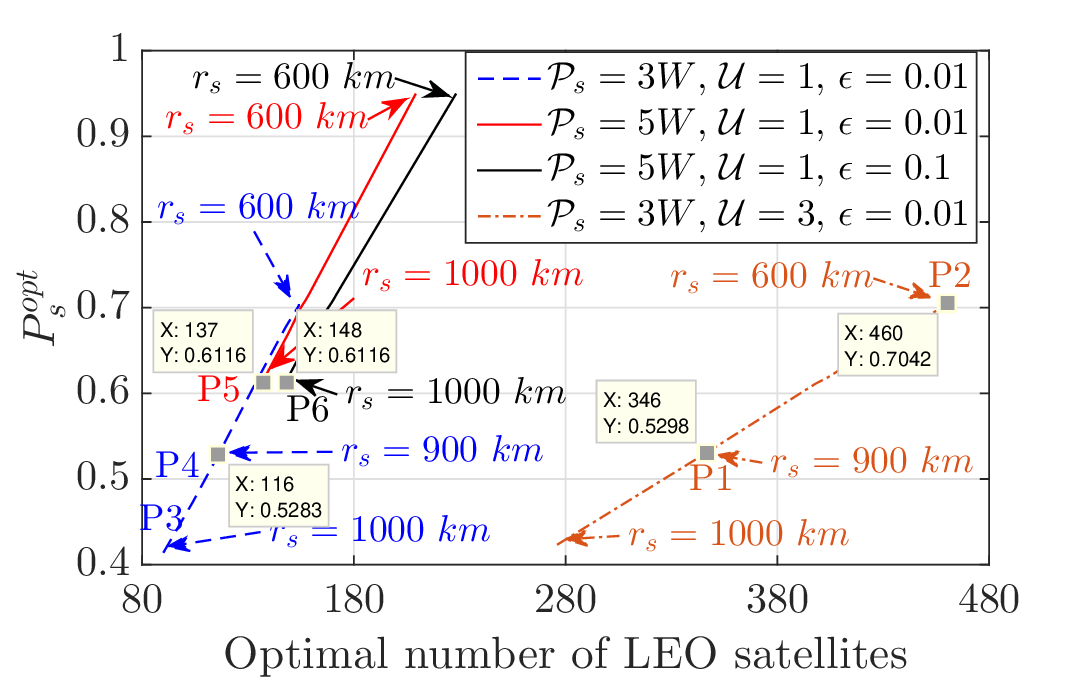}
\caption{Optimal offloading probability vs. optimal number of LEO satellites.}
\label{fig2}
\end{figure}

\section{Conclusion}
In this paper, we investigated traffic offloading problem in integrated LEO satellite-terrestrial network and derived the traffic offloading probability from the terrestrial network to the satellite network based on the instantaneous radio signal strength. We designed the optimal satellite network in which the traffic offloading probability is maximized to control the terrestrial traffic, while the number of LEO satellites being idle is close to zero. Our analyses show that a minimum number of LEO satellites overcomes the limited coverage in a region with low intensity of the BSs and also relieves the traffic congestion in a region that the intensity of the users increases.

\end{document}